\documentclass[lettersize,journal]{IEEEtran}
\usepackage{amsmath,amsfonts}
\usepackage{array}
\usepackage{textcomp}
\usepackage{stfloats}
\usepackage{url}
\usepackage{verbatim}
\usepackage{graphicx}
\usepackage{cite}
\usepackage{charlesmacros}
\usepackage{hyperref}
\usepackage{cleveref}
\usepackage{orcidlink}
\usepackage{tcolorbox}
\usepackage{footnote}

\BeforeBeginEnvironment{assbox}{\savenotes}
\AfterEndEnvironment{assbox}{\spewnotes}

\newtcolorbox{theobox}{colback=green!10, arc=3mm,colframe=red!0, boxrule=0pt}
\newtcolorbox{propbox}{colback=blue!5, arc=3mm,colframe=red!0, boxrule=0pt}
\newtcolorbox{lembox}{colback=blue!5, arc=3mm,colframe=red!0, boxrule=0pt}
\newtcolorbox{assbox}{colback=red!5, arc=3mm,colframe=red!0, boxrule=0pt}
\newtcolorbox{defbox}{colback=red!5, arc=3mm,colframe=red!0, boxrule=0pt}

\BeforeBeginEnvironment{lemma}{\begin{lembox}}
\AfterEndEnvironment{lemma}{\end{lembox}}
\BeforeBeginEnvironment{proposition}{\begin{propbox}}
\AfterEndEnvironment{proposition}{\end{propbox}}
\BeforeBeginEnvironment{theorem}{\begin{theobox}}
\AfterEndEnvironment{theorem}{\end{theobox}}

\newcommand{\lA}{\tilde{{A}}}
\newcommand{\cy}{\widehat{\bm{y}}}
\newcommand{\A}{{A}}
\newcommand{\cA}{\widehat{{A}}}
\newcommand{\cB}{\widehat{{B}}}
\newcommand{\Al}{\widehat{{A}}}

\newcommand{\y}{\bm{y}}
\newcommand{\ly}{\tilde{\bm{y}}}
\newcommand{\X}{{X}}
\newcommand{\x}{{x}}
\newcommand{\lX}{\widetilde{{X}}}
\newcommand{\z}{\bm{z}}
\newcommand{\B}{{B}}

\newcommand{\C}{{D}}
\newcommand{\cC}{\widehat{D}}

\newcommand{\lC}{\tilde{D}}
\newcommand{\lB}{\widetilde{{B}}}

\newcommand{\tL}{\widetilde{L}}
\newcommand{\inner}[2]{\langle #1, #2 \rangle}
\newcommand{\accX}{\lbar{\X}}

\DeclareMathOperator{\dist}{dist}

\begin{document}
\title{Distributed Adaptive Spatial Filtering with Inexact Local Solvers}

\author{
Charles Hovine \orcidlink{0000-0001-6657-1066}, Alexander Bertrand
\orcidlink{0000-0002-4827-8568}

\thanks{This project has received funding from the European Research Council
(ERC) under the European Union's Horizon 2020 research and innovation programme
(grant agreement No. 802895 and No. 101138304) and from the Flemish Government under the
``Onderzoeksprogramma Artifici\"ele Intelligentie (AI) Vlaanderen'' programme. Views and opinions expressed are however those of the author(s) only and do not necessarily reflect those of the European Union or ERC. Neither the European Union nor the ERC can be held responsible for them.}
\thanks{Charles Hovine and Alexander Bertrand are with the STADIUS Center for
        Dynamical Systems, Signal Processing and Data Analytics and with the Leuven.AI institute for Artificial Intelligence at KU Leuven,
        Leuven 3001, Belgium (e-mails: \{charles.hovine,
alexander.bertrand\}@esat.kuleuven.be).}

}


\maketitle

\begin{abstract}
The Distributed Adaptive Signal Fusion (DASF) framework is a meta-algorithm for computing data-driven spatial filters in a distributed sensing platform with limited bandwidth and computational resources, such as a wireless sensor network. The convergence and optimality of the DASF algorithm has been extensively studied under the assumption that an exact, but possibly impractical solver for the local optimization problem at each updating node is available. In this work, we provide convergence and optimality results for the DASF framework when used with an inexact, finite-time solver such as (proximal) gradient descent or Newton's method. We provide sufficient conditions that the solver should satisfy in order to guarantee convergence of the resulting algorithm, and a lower bound for the convergence rate. We also provide numerical simulations to validate these theoretical results.
\end{abstract}

\section{Introduction}

A wireless sensor network (WSN) is a networked collection of sensor nodes equipped with communication and computing capabilities. Each node typically senses a single or multi-channel signal, which can either be transmitted to a fusion center for processing, or be processed \emph{in-network}. 
Spatial filtering is a common processing task in such networks, allowing the extraction of some target signal from the aggregated sensor signals. A spatial filter typically consists in a linear combination of the sensor signals, producing a lower-dimensional output signal which is optimal in some sense, for example, in terms of signal-to-noise ratio, output power, or correlation with a target signal. Common examples include principal component analysis (PCA) \cite{Hotelling1933}, Wiener filtering \cite{wiener1949extrapolation}, canonical correlation analysis (CCA) \cite{Bartlett1941}, linearly constrained minimum variance beamforming (LCMV) \cite{markovich2015optimal} and Max-SNR filtering \cite{van1988beamforming}.

The most straightforward way to compute such filters is to aggregate the sensor data at a fusion center, where an off-the-shelf algorithm can be used to process the data. This convenience comes with several drawbacks. Firstly, the transmission of the sensor nodes' raw sensor signals are likely to become the main energy bottleneck of the system \cite{corke2010environmental}, in particular in networks where multi-hop data relaying is required, resulting in shorter battery life, making remote deployments impractical. This poses a particular challenge for applications that generate continuous streams of high-throughput sensor data, such as audio in acoustic sensor networks \cite{bertrand2010adaptive}, video in CCTV systems or video sensor networks \cite{elharrouss2021review}, biomedical signals in EEG\footnote{Electro-encephalography.} sensor networks \cite{narayanan2020optimal, narayanan2021eeg}, or radio signals in multistatic radars \cite{inggs2014multistatic, inggs2014nextrad}.
Secondly, the fusion center needs to have sufficient computing capabilities to handle the continuous processing of the stream of high-dimensional input signals. This comes with scalability issues, as the increase of the number of nodes or input signals will eventually lead to unrealistic hardware requirements at the fusion center in terms of bandwidth, memory, and compute. 
Finally, the use of a fusion center introduces a single point of failure \cite{Haykin2010}, which can be prohibitive if maintenance actions are difficult or costly, such as in the case of remote deployments. The drawbacks associated with centrally processing the data motivate the use of distributed signal processing algorithms that distribute the computational burden amongst the nodes and favor local processing over data transmission. 

We can identify two classes of distributed signal processing algorithms, associated with two corresponding kinds of distributed datasets. The first kind covers algorithms operating on the observations of a stochastic signal $\bm{y}(t)\in\mathbb{R}^M$, whose observations are distributed across the nodes, such that each node has access to all $M$ entries of $\bm{y}(t)$, but only a subset of the realizations of the random data associated with those entries. This allows each node to locally estimate the covariance structure of the full signal. The optimality criterion can in this case typically be expressed as a sum of local per-node objectives, and the distributed algorithm will usually consist in performing some local processing on the node, and then exchanging some low-dimensional vector containing intermediate estimates of the optimization variables (rather than signal samples). Typical examples of algorithms suited for such datasets are the alternating direction method of multipliers (ADMM) \cite{boyd2011distributed}, consensus and diffusion strategies\cite{lopes2008diffusion,sayed2014adaptation}, and many of the techniques studied by the federated learning community \cite{li2020federated}. 

The second kind of distributed algorithms deals with datasets where the entries (or channels) of $\bm{y}(t)$ are spread across the nodes, such that no individual node can estimate the correlation structure of the network-wide signal $\bm{y}(t)$ (i.e. the signal consisting in the concatenation of all the channels of the per-node signals). Such algorithms typically require that signal samples (rather than low-dimensional parameter vectors) are shared between the nodes, in order to learn the inter-node statistics and to properly evaluate the optimality criterion. Distributed spatial filtering problems or regression/classification problems with distributed features  fit in this second class, and are more difficult to solve if these inter-node statistics are not known a-priori. For streaming data or data sets where the number of observations of $\bm{y}(t)$ is much larger than the number of channels $M$, the typical work-horse distributed methods (ADMM, consensus, diffusion, etc.) do not straightforwardly apply \cite{Kanatsoulis2018}, or result in highly inefficient multi-rate communication schemes with nested iterations\cite{Scaglione2008, Li2011}. Therefore, they are typically solved with ad-hoc problem-specific algorithms, due to the absence of a one-fit-all off-the-shelf solution \cite{zeng2013distributed, Li2011, markovich2015optimal, Fu2018, Bertrand2015b, mota2011distributed, zhang2019distributed, hovine2022maxvar}. \cite{liu2022fedbcd} describes an algorithm to solve such distributed-features problem, but requires a central node (although most of the computational load is distribued amongst the nodes). In addition, it is limited to smooth and unconstrained problems.

The Distributed Adaptive Signal Fusion (DASF) framework \cite{musluoglu2022unified_p1, musluoglu2022unified_p2, hovine2023distributed, hovine2024distributed} aims to bridge the gap by providing a scalable and bandwidth-efficient algorithm that adaptively computes the outputs of data-driven spatial filters satisfying some optimality criterion, in a WSN or other distributed setting with constrained resources. Given an existing solver for the particular optimization problem associated with the filter of interest, the DASF algorithm can adaptively compute the desired filter along with the filtered signal by relying on the exchange  of low-dimensional compressed samples between the nodes, and the repeated computation of the solution of a lower-dimensional version of the original centralized optimization problem at each node. The DASF framework is mostly problem-agnostic, and most traditional linear spatial filtering or estimation problems can be solved by plugging-in the appropriate (and unmodified) solver for the corresponding centralized problem. However, the original DASF framework was shown to converge under the assumption that the solver produces an (near-)exact solution every time it is called.  In practice though, an exact solver might not be available, either due to the nature of the problem (e.g. non-convexity), or simply because obtaining an exact solution would be too expensive or time consuming, hindering the adaptivity of the algorithm. In addition, the convergence of DASF relies on hard-to-check assumptions regarding the parametric continuity of the problem. In this work, we show that an exact solution is not required for optimality, and that a few iterations of an iterative solver are in practice sufficient to ensure convergence to an ``interesting'' filter (i.e. satisfying some relaxed optimality condition such as mere stationarity with regards to the optimality criterion). We are also able to eliminate some of the original convergence assumptions, making the algorithm applicable to a broader set of problems.

The paper is organized as follows. Section \ref{sec:problem} describes the general objective, as well as the family of problems covered by DASF. Section \ref{sec:dasf} gives a brief overview of the original DASF algorithm. Section \ref{sec:inexactalg} describes how DASF can be modified to work with an inexact iterative solver, rather than an exact one. This section  describes our main contribution, including a proof of the convergence of DASF used with such a solver. Section \ref{sec:sims} offers a validation of our theoretical results via numerical simulations. We conclude with a brief discussion in Section \ref{sec:concl}.

\section{Problem Statement}
\label{sec:problem}
We consider a WSN consisting of $K$ nodes, each sensing an $M_k$-dimensional stochastic signal $\y_k(t)$, where $t$ denotes a sample index (usually related to a time index). We denote the $M$-dimensional network-wide sensor signal $\y(t) = [\y_1^T(t),\dots,\y_K^T(t)]^T$ with values in $\R^M$, with $M=\sum_k M_k$. We assume that $\y$ is short-term stationary and ergodic, such that its statistics can be computed over short-term sample batches. The DASF framework assumes that the nodes in the network collaborate to compute an optimal linear spatial filter $\X\in\R^{M\times Q}$ in a bandwidth-efficient manner. Similarly to $\y$, we define the per-node partition of $\X$ as  $\X = [\X_1^T,\dots,\X_K^T]^T$, where we refer to $\X_k$ as the ``local'' filter associated with node $k$. The $Q$-channel filtered signal $\z(t)\triangleq \X^T \y(t)$, and hence the filter $\X$ should satisfy some optimality criterion, which can be expressed as\cite{musluoglu2022unified_p1}

\begin{equation}
    \label{eq:generic_problem}
    \begin{split}
    \X^\star & \in \argmin_{\X\in\mathbb{R}^{M\times Q}} \; \varphi(\X^T \bm{y}(t), \X^T\B)\\
    \st &\forall j\in\mc{J}_I,\; \eta_{j}(\X^T\y(t), \X^T \C_j ) \leq 0,\\
            &  \forall j\in\mc{J}_E, \;\eta_{j}(\X^T\y(t), \X^T \C_j ) = 0
    \end{split}
\end{equation}
where the matrices $\B$ and $\C_{j}$  are
deterministic matrices known by every node, $\varphi$ is a
smooth real-valued function encoding some design objective for the filter
output, $\mc{J}_I$ and $\mc{J}_E$ are the sets of inequality and equality constraints indices, respectively, and the $\eta_{j}$ are smooth
functions enforcing some hard
constraints on the filter outputs and/or filter coefficients. 

In order for the above functions to be real-valued, we assume that they implicitly contain an operator turning the random argument $\X^T\y(t)$ into a deterministic one, by applying, e.g., an expectation operator.  For the purpose of mathematical tractability, we assume for the theoretical analysis of the algorithm that $\y(t)$ is stationary (i.e. that its statistics are independent of time) and ergodic (i.e. its statistics can be estimated using samples averages), and its statistics are assumed to be perfectly estimated at any point in time. However, the expectation operators will in practice be estimated by temporal averages of $\y(t)$. The impact of such an approximation on the algorithm's performance is outside of the scope of this paper, we therefore refer the interested reader to the stochastic optimization literature \cite{kim2015guide} for details.

The peculiar structure of problem \eqref{eq:generic_problem} ensures that $\X$ always appears either in an inner product with $\y(t)$ or in an inner product with some pre-determined matrices ($\B$ or $\C_j$). Most traditional spatial filtering or linear estimation problems satisfy this structure (two examples are given hereafter, and more extensive illustrations can be found in \cite{musluoglu2022unified_p1}). The DASF algorithm will exploit this inner-product structure to achieve a bandwidth reduction. Note that the framework also allows for complex-valued signals, in which case all transpose operators should be replaced with conjugate (a.k.a. Hermitian) transpose operators. Finally, it is noted that we only include a single $\y$ and $\B$-matrix in the loss function $\varphi$, and a single $\C$-matrix in each constraint. This is for the sake of notational convenience, as the DASF framework also allows multiple instances of each (details in \cite{musluoglu2022unified_p1}). 

In order to also cover the treatment of non-smooth objective functions, we also allow the optimality criterion to be expressed as \cite{hovine2023distributed, hovine2024distributed}
\begin{equation}
    \label{eq:generic_problem_ns}
    \begin{split}
    \X^\star & \in \argmin_{\X\in\mathbb{R}^{M\times Q}} \; \varphi(\X^T \bm{y}(t), \X^T\B) + \gamma(\X^T\A)\\
        \st &  \forall  \; k\in \mc{K},\\
            &\forall j\in\mc{J}_I^k,\; \eta_{j}(\X_k^T\y_k(t), \X^T_k \C_{j,k} ) \leq 0,\\
            &  \forall j\in\mc{J}_E^k, \;\eta_{j}(\X_k^T\y_k(t), \X^T_k \C_{j,k} ) = 0
    \end{split}
\end{equation}
where the main differences with \eqref{eq:generic_problem}  are the restriction of the constraints to be per-node block-separable, i.e., the constraints index sets  $\mc{J}_I^k$ and $\mc{J}_E^k$ describe the constraints for a specific node $k$, and the functions $\eta_j$ associated with a particular node $k$ can only depend on the local filter $\X_k$ associated with that node, and the addition of $\gamma$, a possibly non-smooth function. The function $\gamma$ is also required to be per-node block-separable, i.e., there exist matrices $A_k$ and functions $\gamma_k$ such that
\begin{equation}
\label{eq:per_node_block_separable}
\gamma(\X^T\A) = \sum_{k\in\mc{K}} \gamma_k(\X_k^T\A_k),
\end{equation}
which also implies that $\A$ in \eqref{eq:per_node_block_separable} must be a block-diagonal matrix with blocks $\A_k$. This allows for most of the typical regularizers, including the $\ell_1$, $\ell_2$, and the $\ell_{1,2}$ norm. Note that for simplicity, we previously assumed in \cite{hovine2024distributed} that $\gamma$ was a convex function, but this assumption is here relaxed, and $\gamma$ is thus allowed to be non-convex.

We refer to the parametric optimization problem defined by \eqref{eq:generic_problem} as $\mathbb{P}_S(\y(t), \B, \mc{D})$, and $\mathbb{P}_{NS}(\y(t),\B, \mc{D}, \A)$ for \eqref{eq:generic_problem_ns}. Here the set $\mc{D}$ denotes the collection of all the  matrices $D_j$ or $D_{j,k}$.
These generic problem formulations encompass a wide range of well-known spatial filtering problems, including PCA, CCA, LCMV, and many more \cite{musluoglu2022unified_p1}. For example, we can express a PCA filter as
\begin{equation}
    \begin{split}
        \min_\X &\quad -\tr{\X^T\E{\y\y^T}\X}\\
                &\st  \X^T\X =  I,
\end{split}
\end{equation}
where $\E{\cdot}$ denotes the expectation operator. A $\C$ matrix hides in the constraint, which can equivalently be written 
\begin{equation}
    \X^TDD^T\X = I,\quad \C=I.
\end{equation}
This might seems superfluous, but $\C$ actually plays an important aglorithmic role, as its presence will allow the local sub-problems solved during the application of DASF to keep the same structure ${P}_S(\y(t), \B, \{\C\})$ as the original one, but where $\C\neq I$. The $\ell_{1,2}$-regularized multi-channel Wiener filter \cite{vaseghi1996wiener} is a non-smooth example that fits in the class of problems defined by \eqref{eq:generic_problem}-\eqref{eq:generic_problem_ns}:
\begin{equation}
    \label{eq:wiener}
        \begin{split}
                \min_{\X} &\; \E{\norm{\X^T\y-\bm{d}}_F^2} + \sum_{k}
                \norm{\X_k}_{F}\\
        \end{split}
\end{equation}
where $\norm{\cdot}_F$ denotes the Frobenius norm, and
$\bm{d}(t)$ is a known multichannel signal taking values in $\mathbb{R}^{Q}$. The filter produces an estimate of the signal $\bm{d}(t)$ from the measurements $\bm{y}(t)$, and the
non-smooth regularization term encourages nodes that do not significantly contribute to
the objective to set their filters to zero. Note that $\X$ does not appear in an inner product in the second term of \eqref{eq:wiener}, yet it fits \eqref{eq:generic_problem}-\eqref{eq:generic_problem_ns} when explicitly writing the Frobenius norm as $\norm{\X_k^T \A_k}_F$ with $\A_k=I$.

\section{DASF Overview}
\label{sec:dasf}
\begin{figure*}
        \usetikzlibrary{arrows.meta}
\centering\hfill
\tikzstyle{nonupdate}=[draw=none, circle, minimum size=10pt, fill=blue!20]
\tikzstyle{nonupdate_in}=[rectangle, dashed, draw=black, fill=blue!5]
\tikzstyle{update}=[draw=none, circle, minimum size=10pt, fill=orange!20]
\tikzstyle{arr}=[-{Triangle[scale=1.4]}]
\tikzstyle{arr2}=[{Triangle[scale=1.4]}-]
\subfigure[Data aggregation and local solution]{
    \begin{tikzpicture}
    \def\dist{2.5}
    \node[nonupdate] (k1) at (0,0) {1};
    \node[nonupdate] (k2) at (-\dist,0) {2};
    \node[nonupdate] (k3) at (-\dist,-\dist) {3};
    \node[update] (q) at (0,-\dist) {$q$};

    \draw[arr] (k1) -- (q) node[midway, right] {$\cy_1$};
    \draw[arr] (k2) -- (q) node[midway, above right] {$\cy_2$};
    \draw[arr] (k3) -- (q) node[midway, below] {$\cy_3$};

    \node[above of=k1, node distance=15pt, anchor=south, nonupdate_in] (k1t) {$\cy_1 = \X^{iT}_1\y_1$};
    \node[above of=k2, node distance=15pt, anchor=south, nonupdate_in] (k2t) {$\cy_2 = \X^{iT}_2\y_2$};
    \node[above of=k3, node distance=15pt, anchor=south, nonupdate_in] (k3t) {$\cy_3 = \X^{iT}_3\y_3$};

    \draw[dashed] (k1) to (k1t.south);
    \draw[dashed] (k2) to (k2t.south);
    \draw[dashed] (k3) to (k3t.south);

    \node[above right of=q, anchor=west, node distance=40pt, rectangle, draw=black, dashed, fill=orange!5, align=center] (vecq) {$
            \tilde{\y}_q=\begin{bmatrix}
                \y_q\\
                \cy_1\\
                \cy_2\\
                \cy_3
            \end{bmatrix},
         \tilde{\X}_q^\star=\begin{bmatrix}
                \X_q\\
                \lX_1\\
                \lX_2\\
                \lX_3
            \end{bmatrix}$\\~\\
{ solution of }$\mathbb{P}(\ly^i,\lB^i, \tilde{\mc{D}}^i, \lA^i)$

    };

    \draw[dashed] (q) to ($(vecq.west)+(0,0)$);

\end{tikzpicture}
}\hfill
\subfigure[Parameters update]{
    \begin{tikzpicture}
    \def\dist{2.5}
    \node[nonupdate] (k1) at (0,0) {1};
    \node[nonupdate] (k2) at (-\dist,0) {2};
    \node[nonupdate] (k3) at (-\dist,-\dist) {3};
    \node[update] (q) at (0,-\dist) {$q$};

    \draw[arr2] (k1) -- (q) node[midway, right] {$\lX_1$};
    \draw[arr2] (k2) -- (q) node[midway, above right] {$\lX_2$};
    \draw[arr2] (k3) -- (q) node[midway, below] {$\lX_3$};

    \node[above of=k1, node distance=15pt, anchor=south, nonupdate_in] (k1t) {$\X^{i+1}_1 = \X^{i}_1\lX_1$};
    \node[above of=k2, node distance=15pt, anchor=south, nonupdate_in] (k2t) {$\X^{i+1}_2 = \X^{i}_2\lX_2$};
    \node[above of=k3, node distance=15pt, anchor=south, nonupdate_in] (k3t) {$\X^{i+1}_3 = \X^{i}_3\lX_3$};
    \node[above right of=q, node distance=15pt, anchor=west, rectangle, draw=black, dashed, fill=orange!5] (vecq) {$\X^{i+1}_q = {\X}_q$};

    \draw[dashed] (k1) to (k1t.south);
    \draw[dashed] (k2) to (k2t.south);
    \draw[dashed] (k3) to (k3t.south);
    \draw[dashed] (q) to ($(vecq.west)+(0,0)$);
\end{tikzpicture}
}\hfill
        \caption{Overview of a single iteration of the DASF Algorithm in fully-connected networks, where node $q$ is the updating node. The matrices $ \B_k, \C_{j,k}$ and $\A_k$ are omitted for readability, but their treatment is the same as $\y_k$.}
        \label{fig:dasf_steps}
\end{figure*}
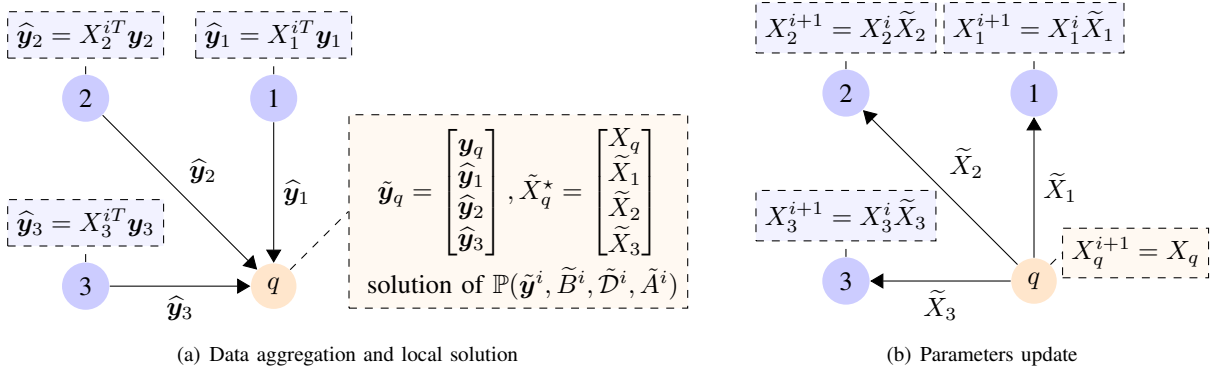

In this section, we briefly review the DASF algorithm. We refer to \cite{musluoglu2022unified_p1} for more details and illustrative examples.

The DASF algorithm collaboratively updates the estimate of the optimal filter
$\X^\star$ and tracks the filtered signal $\z$ by solving a local ``miniature''
version of \eqref{eq:generic_problem} or \eqref{eq:generic_problem_ns} at a particular so-called \emph{updating node} whose role is taken by a different node at every iteration. 
This makes the DASF algorithm ``plug-and-play''\footnote{ A Matlab and Python toolbox implementing this plug-and-play functionality is available in \cite{dasf_toolbox}.} as it only requires a solver for the centralized problem \eqref{eq:generic_problem} or \eqref{eq:generic_problem_ns}, which is directly re-used as a local ``sub-solver'' in the compressed local problems in the different iterations of the distributed algorithm \cite{musluoglu2022unified_p1}.

As the focus of this paper is the local sub-solver used by DASF, which is largely independent of the rest of the algorithmic framework, we can afford to limit our description of the algorithm to fully-connected networks.  It is important to note that this is purely for intelligibility and notational convenience, but without loss of generality. The results described in this paper
can be straightforwardly extended to arbitrary topology networks, as in \cite{musluoglu2022unified_p1, musluoglu2022unified_p2} (for the smooth case) and \cite{hovine2024distributed} (for the non-smooth case) and the convergence proof in Section \ref{sec:convergence} is kept sufficiently generic to also cover the case of arbitrary-topology networks. As a result, the extension of DASF to arbitrary topologies that is detailed in \cite{musluoglu2022unified_p1} and \cite{hovine2024distributed} straightforwardly applies to the inexact version of DASF introduced in this paper, yet it is omitted here to reduce overlap and avoid the extra layer of notational complexity that it would incur.

An iteration $i$ of the DASF algorithm is fully defined by the current updating node index $q^i$ (we often drop the iteration index when the iteration is clear from the context) and the current estimate of the filter $\X^i$, and consists of the following three steps (after a random initialization of $\X^0$):
\begin{description}
        \item[1. Data Aggregation] {
                        At the beginning of a new iteration $i$, a new {updating node} $q$
                        is selected. Each node $k$ collects a new batch of $N$ samples of $\y(t)$ for $t=iN,\dots(i+1)N-1$, and
                        compresses these into $Q$-dimensional\footnote{Here we assume that $Q<M_k$, otherwise there is no compression possible at node $k$. Nodes for which $M_k\leq Q$ can add their channels $y_k$ to the channels of a neighboring node to create a virtual ``super-node'' within the DASF algorithm. Note that in multi-hop network topologies, a bandwidth reduction can even be achieved if $Q>M_k$ in all nodes  (when compared to relayed data aggregation). We refer to \cite{musluoglu2022unified_p1} for more details.} samples according to
                        \begin{equation}
                                \label{eq:compressed_data}
                                \cy_k^i \triangleq \X_k^{iT}\y_k.
                        \end{equation}
                        The compressed $N$-samples batch is then transmitted to the updating node, along with the compressed matrices 
                        \begin{align}
                                \label{eq:compressed_data_2}
                                \cC_{j,k} &\triangleq \X_k^{iT}\C_{j,k}\\
                                \cA_k^i &\triangleq X_k^{iT}\A_k    \\
                                \cB_k^{iT} &\triangleq \X_k^{iT}\B_{k}
                        \end{align}
                        (typically of negligible size compared to the transmission of \eqref{eq:compressed_data} when the sample batch size $N$ is large). Here, $\B_k$ is the block of $\B$ associated with the block $\X_k$ of $\X$. In order to offer a unifying description of the algorithm for the smooth and non-smooth case, we also define the same relationship between $\C_{j,k}$ and $\C_j$ (i.e. in the smooth case $\C_{j,k}$ is the block of $D_j$ associated with $\X_k$ in the product $X^TD_j$), and in order to be consistent with the definitions of problems \eqref{eq:generic_problem} and \eqref{eq:generic_problem_ns}, we require that $\C_{j,k}=0$ if $j\notin \mc{J}_{E}^k\cup \mc{J}_{I}^k$ in the non-smooth case (i.e. block-separable constraints), and $\A_k=0$ in the smooth-case (i.e. smooth objective function). 
                }
        \item[2. Local Solution] {
                        The updating node constructs a \emph{local} view $\tilde{\y}_q$ of $\y$ by
                        concatenating  the received signals with its own sensor
                        signals such that
                        \begin{equation}
                                \label{eq:local_data}
                                \ly^i = \begin{bmatrix}
                                        \y_q^T & \cy_1^{iT} & \cdots & \cy_{q-1}^{iT} & \cy_{q+1}^{iT} & \cdots &
                                        \cy_K^{iT}
                                \end{bmatrix}^T.
                        \end{equation}
                        Similarly, it constructs
                        \begin{equation}
                                \label{eq:local_data_2}
                                \begin{split}
                                        \lA^i &= \blkd(\A_q, \cA_1^i, \ldots,
                                        \cA_{q-1}^i, \cA_{q+1}^i, \ldots, \cA_K^i),\\
                                        \lB^i &= \begin{bmatrix}
                                                \B_q^T & \cB_1^{iT} & \cdots & \cB_{q-1}^{iT} & \cB_{q+1}^{iT} & \cdots &
                                                \cB_K^{iT}
                                        \end{bmatrix}^T,
                                        \textrm{ and }\\
                                        \lC^i_j &= \begin{bmatrix}
                                                \C_{j,q}^T & \cC_{j,1}^{iT} & \cdots & \cC_{j,q-1}^{iT} & \cC_{j,q+1}^{iT} & \cdots &
                                                        \cC_{j,K}^{iT}
                                        \end{bmatrix}^T.
                                        \end{split}
                                \end{equation}
                                The updating node then obtains the local solution ${\lX}^\star$ by solving $\mathbb{P}(\ly^i,\lB^i, \tilde{\mc{D}}^i, \lA^i)$ (which here refers to either $\mathbb{P}_S$ or $\mathbb{P}_{NS}$) for the received $N$-sample batch, where, simirlarly to $\mc{D}$, $\tilde{\mc{D}}^i$ denotes the collection of $\lC^i_j$, which can be viewed as a low-dimensional instance of the original optimization problem. Note that the solver for the network-wide problem $\mathbb{P}(\y,\B, \tilde{\mc{D}}, \A)$ can thus also be used as-is to solve the ``compressed'' problem $\mathbb{P}(\ly^i,\lB^i, \tilde{\mc{D}}^i, \lA^i)$ at updating node $q$ \cite{musluoglu2022unified_p1}. 
                        }
                \item[3. Parameters Update]
                        The updating node partitions the local solution $\lX^\star$ as\footnote{Note that $\X^\star_q\in\mathbb{R}^{M_q\times Q}$ and $\lX^\star_k\in\mathbb{R}^{Q\times Q}$ for all other $k\neq q$. This follows from the dimensions of the partitioning in \eqref{eq:local_data}.}
                        \begin{equation}
                                \label{eq:block_struct_local_filter}
                                \lX^\star = [X_q^{\star T}, \lX_1^{\star T}, \ldots, \lX_{q-1}^{\star T}, \lX_{q+1}^{\star T}, \ldots, \lX_K^{\star T}]^T.
                        \end{equation}
                        It then updates its own block of the filter as 
                        \begin{equation}
                                \label{eq:update_rule_q}
                                \X_q^{i+1}\leftarrow X_q^{\star},
                        \end{equation}
                        and transmits each $\lX_k$ to its corresponding node, which in turn updates its local filter as 
                        \begin{equation}
                                \label{eq:update_rule_k}
                                \X_k^{i+1}\leftarrow \X_k^i\lX_k^{\star}.
                        \end{equation}
        \end{description}
        The three steps are illustrated in Figure \ref{fig:dasf_steps} and the full procedure is described by Algorithm \ref{alg:fc}.
        \SetKwFor{Node}{At node}{}{}
\SetKwFor{Loop}{loop}{}{}
\begin{algorithm}[t]
        \Begin{
                $i\gets 0$, $q\gets 1$, Randomly initialize $\X^0$\\
                \Loop{}{
                        \For{$k\in\mc{K}\smallsetminus\{q\}$}{
                                \Node{$k$} {
                                        Collect a new batch of $N$ samples of
                                        $\y_k(t)$ and send the compressed
                                        samples
                                        $\cy_k^i(t)=\X_k^{iT}\y_k(t)$
                                        along with $\Al^i_k=\X^{iT}_kA_k$, $\cB^i_k=\X^{iT}_k\cB_k$ and
                                        $\cC_k^i = \X^{iT} \C_k$ to node $q$.
                                }
                        }
                        \Node{$q$}{
                                Obtain $\lX^{\star}$ by solving and selecting any solution of
                                $\mathbb{P}(\ly^i,\lB^i, \tilde{\mc{D}}^i, \lA^i)$ .\\                                Extract $\X_q^\star$ and the $\lX_k^{\star}$'s
                                from ${\lX}^{\star}$ according to \eqref{eq:block_struct_local_filter}.\\
                                $\X_q^{i+1}\gets \X_q^{\star}$\\
                                \For{$k\in\mathcal{K}\smallsetminus\{q\}$}{

                                        Send $\lX_k^{\star}$ to node $k$.\\
                                        \Node{$k$}{
                                                $\X^{i+1}_k \gets
                                                \X^{i}_k\lX_k^{\star}$\\
                                        }
                                }

                        }
                        $i\gets i+1$, $q\gets (q + 1) \mod K$

                }
        }
        \caption{(NS-)DASF algorithm in fully-connected networks. Implementation available in \cite{dasf_toolbox}}
        \label{alg:fc}
\end{algorithm}

        As a different batch of $N$ samples is used at each iteration, the DASF algorithm produces the filtered signal 
        \begin{equation}
                \z^{i+1}\triangleq\X^{i+1T}\y=\sum_k \X_k^{i+1T}\y_k=\lX^{\star T}\ly^{i} 
\end{equation}
for each $N$-samples block, while at the same time improving the estimate of the optimal spatial filter $\X^\star$, such that each new block of the filtered signal is closer to the desired filtered signal (under the stationarity assumption). In other words, DASF acts as a time-recursive block-adaptive filter that continually adapts itself over time to the (possibly changing) statistics of $\y(t)$ \cite{musluoglu2022unified_p1}. Note that the last equality implies that each updating node can locally produce the new estimate $\z^{i+1}$ without additional data exchange.

        Provided that the local solution is obtained using an exact solver, the DASF algorithm and its extension to arbitrary network topologies have been shown to converge to a stationary point of the original problem (under some technical conditions, details omitted) \cite{musluoglu2022unified_p2, hovine2024distributed}.

\section{DASF with Inexact Local Solvers}
\label{sec:inexactalg}


The algorithm described in the previous section might in practice be difficult to implement, as the computational cost to find a global minimizer of $\mathbb{P}(\ly^i,\lB^i, \tilde{\mc{D}}^i, \lA^i)$ at the updating node might be prohibitive, requiring either too much computations or too much memory. In addition, an exact global solver for the problem of interest might not even exist. For example, many solvers for non-convex problems often produce a stationary point or local optimum rather than a global optimum. In this section, we show how the local solution step can be significantly relaxed to allow for common iterative solvers, not necessarily converging to optimal values, while still guaranteeing convergence to a stationnary point of the centralized problem \eqref{eq:generic_problem} or \eqref{eq:generic_problem_ns}.

In what follows, we express \eqref{eq:generic_problem}-\eqref{eq:generic_problem_ns} in the equivalent form 
\begin{equation}
    \label{eq:compact_problem}
    \min_{\X} L(\X)
\end{equation}
where 
\begin{equation}
    \label{eq:sum_cost}
    L(\X)\triangleq \varphi(\X^T \bm{y}(t), \X^T\B) + \gamma(\X^T\A) +\delta_{\mc{X}}(\X)
\end{equation}
with $\delta_{\mc{X}}(\cdot)$ the indicator function of some set $\mc{X}$ described by the equality and inequality constraints, i.e., $\delta_{\mc{X}}(\X) = 0$ if $\X\in\mc{X}$ and $\delta_{\mc{X}}(\X) = \infty$ otherwise\footnote{$L$ is thus not a real-valued function, but takes its values in the \emph{extended} real number line $\overline{\R}\triangleq\R\cup\{-\infty,\infty \}$.}. In the case of the generic problem \eqref{eq:generic_problem} we have
\begin{equation}
    \begin{split}
        \mc{X}\triangleq \{ 
            \X \;|\; &\forall j\in\mc{J}_I,\; \eta_{j}(\X^T\y(t), \X^T \C_j ) \leq 0,\\
                     &\forall j\in\mc{J}_E, \;\eta_{j}(\X^T\y(t), \X^T \C_j ) = 0\},
        \end{split}
    \end{equation}
    and in the non-smooth case \eqref{eq:generic_problem_ns},
\begin{equation}
    \begin{split}
        \mc{X}\triangleq \{ 
            \X \;|\; \forall k,\; &\forall j\in\mc{J}_I^k,\; \eta_{j}(\X^T_k\y_k(t), \X_k^T \C_{j,k} ) \leq 0,\\
                                  &\forall j\in\mc{J}_E^k, \;\eta_{j}(\X^T_k\y_k(t), \X^T_k \C_{j,k} ) = 0\}.
        \end{split}
    \end{equation}

    From the update rules \eqref{eq:update_rule_q} and \eqref{eq:update_rule_k}, one could show that solving $\mathbb{P}(\ly^i,\lB^i, \lC^i, \lA^i)$ in the space where $\lX$ lives is equivalent to solving $\mathbb{P}(\y,\B, \C, \A)$ in the original space of $\X$ with the additional constraints that 
    \begin{equation}
        \label{eq:param_constraint}
        \forall k\neq q,\;\exists \lX_k\in\mathbb{R}^{Q\times Q}:\;\X_k=\X^i_k \lX_k,
    \end{equation}
    or equivalently
    \begin{equation}
        \forall k\neq q,\; \X_k \in \range \X_k^i,
    \end{equation}
    with $\range$ denoting the range or column space operator. 
    To each local $\lX$ is associated a unique $\X$ in the original problem space, a fact that we can express as the linear relationship
    \begin{equation}
        \label{eq:local_to_global}
        \X = C^i_q \lX 
    \end{equation} 
    where $C_q^i\in\R^{M\times (M_q+(K-1)Q)}$ is a particular block-matrix performing the required permutations and multiplications by the blocks $\X^i_k$ of the blocks $\lX_k$ of $\lX$ to obtain the transformation uniquely described by \eqref{eq:param_constraint}.  
    More specifically, we define
    \label{apx:struct_c}
    \begin{equation}
        \begin{split}
            \Theta_{-q}(X) &\triangleq \blkd\left(\X_1, \dots,\X_{q-1}\right)\textrm { and, }\\
            \Theta_{+q}(X) &\triangleq
           \blkd\left(\X_{q+1}, \dots,\X_{K}\right)
        \end{split}
    \end{equation}
    from which we define the linear matrix-to-matrix map
    \begin{equation}
        \label{eq:def_Cq}
            C_q(\X) = \left[\def\arraystretch{1.5}\begin{array}{c | c |c }
                O & \Theta_{-q}(\X) & O\\\hline
                {I}_{M_q} & O & O\\\hline
                O & O & \Theta_{+q}(\X)
        \end{array}\right]
    \end{equation}
    where the $O$'s denote all-zero matrices of approriate dimensions. This allows us to formally express
    \begin{equation}
        \label{eq:def_Cqi}
        C^i_q \triangleq C_q(\X^i),
    \end{equation}
    leading in particular to 
    \begin{equation}
        \label{eq:first_it}
        \X^i = C^i_q [\X^{iT}_q, I_Q, \dots, I_Q]^T.
    \end{equation}
    
    It is important to note that there is a linear relationship between the \emph{local} variable $\lX$ and the corresponding point $\X$ in the \emph{global} space, and that this relationship depends on the current iteration $i$ (via the previous estimate of the solution $\X^i$) and the current updating node $q$. This relationship \emph{itself} linearly depends on $X^i$ via \eqref{eq:def_Cqi}. The linearity, and hence continuity of the map $C_q$ is one of the key properties used to show convergence, and replacing $C_q$ by any other continuous map would lead to the same convergence result (altough optimality would not anymore be guaranteed). Although we haven't described the algorithm for arbitrary network topologies, a similar linear map as in \eqref{eq:def_Cq} can be defined for that case (see \cite{musluoglu2022unified_p1, hovine2024distributed}), which implies that the convergence proof in Section \ref{sec:convergence} will generalize to these arbitrary topologies as well.

    Using the notation \eqref{eq:sum_cost} and \eqref{eq:local_to_global}, we can  now compactly express the local problem at iteration $i$ as
    \begin{equation}
        \label{eq:loc_prob_compact}
        \min_{\lX} L(C_q^i\lX).
    \end{equation}
    We note that --by construction-- \eqref{eq:loc_prob_compact} has the same solutions as the local problem $\mathbb{P}(\ly^i,\lB^i, \tilde{\mc{D}}^i, \lA^i)$. While we need \eqref{eq:loc_prob_compact} in the proof, in practice it is better to directly solve $\mathbb{P}(\ly^i,\lB^i, \tilde{\mc{D}}^i, \lA^i)$ instead of \eqref{eq:loc_prob_compact} as the former directly inherits the same structure as the original centralized problem \eqref{eq:generic_problem}-\eqref{eq:generic_problem_ns}, and hence the original solver for \eqref{eq:generic_problem}-\eqref{eq:generic_problem_ns} can be re-used.

    \subsection{Useful concepts from variational analysis}
    We introduce here a few definitions from variational analysis that are essential to our next developments. This section is purposedly kept short, and we refer the reader to the reference works of Clarke \cite{clarke1990optimization} and Rockafellar \cite{rockafellar2009variational} for an in-depth treatment of the subject.

    \paragraph{Subgradients}
    Following \cite[Definition 8.3]{rockafellar2009variational}, we define the set of regular subgradients of $L$ at $\accX$ (where $L(\accX)$ is finite) as
    \begin{multline}
        \hat{\partial}L(\accX)\triangleq \Big\{V\in \R^{M\times Q} \;|\;\\\liminf_{\X\to\bar{\X} \atop \X\neq\bar{\X}} \frac{L(\X)-L(\accX)-\inner{V}{\X-\accX}_F}{\norm{\X-\accX}}\geq 0 \Big\},
    \end{multline}
    and the set of general subgradients of $L$ at $\accX$ as\footnote{\eqref{eq:def_subgradient} can be read as ``the set of elements $V$ such that there exist sequences $(\X^j)_j$ converging to $\accX$, $(L(\X^j))_j$ converging to $L(\accX)$ and $(V^j)_j$ converging to $V$, such that $V^j$ is a regular subgratient of $L$ at $X^j$''.} 
    \begin{multline}
        \label{eq:def_subgradient}
        \partial L(\accX)\triangleq \{V\in\R^{M\times Q}\;|\; \\
        \exists \X^j\to\accX, L(\X^j)\to L(\accX), V^j\in\hat{\partial}{L}(\X^j)\to V\},
    \end{multline}
    which we refer to simply as ``subgradients'' from hereon.
    This definition of a subgradient does not require any structural property such as convexity, smoothness, or even continuity from $L$, and extends the usual definition of subgradients from convex analysis. Indeed, \eqref{eq:def_subgradient} generalizes the familiar notion of subgradient of a convex function. Note  again that the developments in \cite{hovine2024distributed} assumed convexity for the non-smooth part of the objective in \eqref{eq:generic_problem_ns}, which is not required here.
    \paragraph{Local Optimality}
    If $L:\R^{M\times Q}\mapsto \overline{\R}$ is lower-semicontinuous\footnote{A function $L$ is lower-semicontinuous if for any $\alpha$, its level sets $\{\X\;|\; L(\X) \leq \alpha\}$ are closed.} and $\X^\star$ is a local minimum of $L$, then \cite[Theorem 8.15]{rockafellar2009variational}
    \begin{equation}
        \label{eq:optimality_cond}
        0\in\partial L(\X^\star).
    \end{equation}
    Any point $\X^\star$ such that $L(\X^\star)$ is finite and satisfying \eqref{eq:optimality_cond} is called a \emph{stationary point} of $L$. 
    This definition reduces to the usual ``KKT'' optimality conditions \cite{karush1939minima, kuhnnonlinear} in the smooth case. Indeed, if 
    \begin{equation}
        L(\X)\triangleq f(\X) + \delta_{\mc{X}}(\X)
    \end{equation}
    where $f$ is a smooth function, and $\mc{X}$ is the closed set $\{\X\;|\; g(\X)\leq0,\; h(\X)= 0\}$ where $h$ and $g$ are smooth functions, then we have at any point $\X$ where $L(\X)$ is finite (and thus $\X\in\mc{X}$) \cite[Proposition 4.58]{Royset2021}
    \begin{align}
        \partial L(\X) &= \nabla f(\X) + \partial\delta_{\mc{X}}(\X).
    \end{align}
    If $\nabla g(\X)$ and $\nabla h(\X)$ are linearly independent (this is stricter than required, but is only meant for illustrative purposes), then \cite[Example 4.49]{Royset2021}
    \begin{multline}
        \partial \delta_{\mc{X}}(\X) = \{\lambda_g \nabla g(\X) + \lambda_h \nabla h(\X) \;|\; \\\lambda_h \in\R,\; \lambda_g = 0 \textrm{ if } g(\X) < 0,\textrm{ else } \lambda_g \geq 0\}.
    \end{multline}
    Defining the \emph{Lagrangian}
    \begin{equation}
        \mathfrak{L}(\X,\lambda_g, \lambda_h)\triangleq f(\X) + \lambda_g g(\X) + \lambda_h h(\X),
    \end{equation}
    we have for a feasible point $\X$ the equivalence
    \begin{multline}
        0\in\partial L(\X) \Leftrightarrow\\ \exists \lambda_g = 0 \textrm{ if } g(\X) < 0,\textrm{ else } \lambda_g \geq 0, \lambda_h:\\ \nabla_{\X} \mathfrak{L}(\X,\lambda_g, \lambda_h)=0.
    \end{multline}
    This equivalence allows us to offer a unifying treatment of the smooth and non-smooth cases via the compact problem formulation \eqref{eq:compact_problem} (or \eqref{eq:loc_prob_compact} for the local problem at the updating node).

    \subsection{DASF algorithm with inexact solver}

    In what follows, we define the local objective at iteration $i$
    \begin{equation}
        \label{eq:local_objective}
        \tL^i_q: \lX\mapsto  L(C_q^i\lX),
    \end{equation}
    which we use to construct the following definition of an \emph{inexact iterative solver}:
    \begin{defbox}
    \begin{definition}[Inexact iterative solver]
        \label{def:inexact_solver}
        Given some current estimate $\X^i$, an inexact iterative local solver running on node $q$ produces a sequence of iterates
        $(\lX^{i,j})_{j\in\mathbb{N}}$ such that for each iteration $i$
        \begin{enumerate}
            \item  \begin{equation}\label{eq:first_prop_ls}\tag{C1} {\lX}^{i,0}=[\X^{iT}_q, {I}_Q,\dots, {I}_Q]^T,\end{equation}
            \item there is some $R^i>0$ such that for any $j$, 
                \begin{multline}
                    \label{eq:second_prop_ls}
                    \tag{C2}
                    \tL_q^i(\lX^{i,j})-\tL_q^i(\lX^{i,j+1})\\\geq
                    R^i\norm{\lX^{i,j}-\lX^{i,j+1}}^2_F,
                \end{multline}
            \item there is some $c^i>0$ such that for any $j$ there is some $W\in\partial \tL_q^i(\lX^{i,j+1})$ such that
                \begin{equation}
                    \label{eq:third_prop_ls}
                    \tag{C3}
                    c^i\norm{\lX^{i,j}-\lX^{i,j+1}}_F\geq \norm{W}_F.
                \end{equation}
            \item In addition, the sequences of constants $R^i$ and $c^i$ must be such that
                \begin{gather}
                    \label{eq:fourth_prop_ls}
                    \tag{C4}
                    \liminf_{i\to\infty} R^i > 0\quad\textrm{and}\quad\limsup_{i\to\infty} c^i <\infty.
                \end{gather}
        \end{enumerate}
    \end{definition}
\end{defbox}
Flowing directly from \eqref{eq:first_it}, condition \eqref{eq:first_prop_ls} ensures consistency with the current global estimate of the solution $\X^i$, and can be enforced as long as the solver accepts an arbitrarily chosen starting point. Condition \eqref{eq:second_prop_ls} is often referred to as a \emph{sufficient decrease} condition \cite{beck2017first, themelis2018proximal}. It ensures that an update of the optimization variable results in a corresponding decrease of the cost, and produces a feasible point (as, by the definition of the  indicator function $\delta_{\mc{X}}$, the cost would otherwise have infinite value in the case of a non-feasible point). Condition \eqref{eq:third_prop_ls} ensures that progressively smaller updates of the optimization variable result in a higher degree of stationarity (the distance of  the set $\partial \tL(\lX^{j+1})$ to 0 plays here the role of optimality criterion, and summarizes how ``close'' we are to a stationary point). This third condition is crucial, as without it, the procedure could wander away from a stationary point during the first few steps, but still eventually converge to some stationary point, in which case a single step of the solver would not be guaranteed to improve on the current iterate. The last requirements on $R^i$ and $c^i$ in condition \eqref{eq:fourth_prop_ls} ensure that those constants do not become arbitrarily small or large as the algorithm progresses, as that would render the two previous conditions meaningless. Those constants are usually related to a parameter of the solver that can be controlled, a lower or upper bound can therefore be enforced explicitly. 

Although these requirements might seem restrictive, they have been shown to hold for many commonly used descent methods \cite{attouch2013convergence}, if some additional technical conditions are satisfied. Here are some common examples, some of which are further elaborated on in Appendix \ref{apx:common_solvers}:
    \begin{description}
        \item[Gradient Descent] (see  Appendix \ref{apx:common_solvers}, and \cite{nocedal1999numerical,attouch2013convergence}) provided that the problem is unconstrained and the objective has Lipschitz gradients.
        \item[Projected/Proximal Gradient Descent] (see \cite{beck2017first, attouch2013convergence}) provided that the smooth part of the objective has Lipschitz gradients.
        \item[Newton's Method] (see Appendix \ref{apx:common_solvers}) provided that the Hessian is positive definite and bounded, and that the objective has Lipschitz gradients.
        \item[Regularized Exact Solver] (see Appendix \ref{apx:common_solvers}) By definition, an exact solver obtains a solution in a single iteration and the properties of Definition \ref{def:inexact_solver} are almost satisfied. If there are multiple solutions, some mechanism is required to ensure that the solver does not ``jump'' from one solution to the next, by for example explicitly penalizing the distance from the starting point. 
        \item[Power Method] (see Appendix \ref{apx:common_solvers}) An intuitive argument is that the power method can be expressed as an instance of projected gradient descent applied to a specific constrained quadratic problem. 
    \end{description}
    \begin{remark}[Comparison with the original DASF algorithm]
        An important, mostly theoritical, difference with the original DASF and NS-DASF algorithms, is the removal of the \emph{well-posedness} assumption on the global problems. Indeed, it was assumed in those algorithms that the solution set of the optimization problem varied ``smoothly'' with regards to the problem's parameters ($\y(t),\B, \mc{D}, \A$). This assumption can be hard to check, and even not hold in some cases. Property \eqref{eq:second_prop_ls}, which is not satisfied by an exact solver, allows us to remove this assumption. See \cite{musluoglu2022unified_p1, musluoglu2022unified_p2, hovine2024distributed} for details on the assumption.
    \end{remark}

    The inexact version of the  DASF Algorithm simply consists in selecting $\lX^{\star}\triangleq \lX^{i,n_i}$ as the solution of an inexact solver applied to the local problem at node $q$  at iteration $i$ of the DASF algorithm, which we denote 
    \begin{equation}
        \label{eq:local_problem}
        \min_{\lX} \tL^i_q(\lX),
    \end{equation}
    where the number of iterations $n_i>0$ of the inexact solver can be chosen arbitrarily. 

    \subsection{Convergence and Optimality}
    \label{sec:convergence}

    We first describe a general and easily satisfied assumption on $L$.
    \begin{assbox}
    \begin{assumption}
        \label{ass:main}
        $L$ is continuous over its domain\footnote{The domain of a function $L$ is the set of points for which  $L(\X)<\infty$.}, and it has compact sublevel-sets.
    \end{assumption}
\end{assbox}
    A sufficient condition is to require that $\varphi$ and the $\eta_j$ are continuous, $\gamma$ is continuous over its domain, and either $\gamma$ and $\varphi$ or at least one of the $\eta_j$ have bounded sublevel sets. 

    We show the convergence of this scheme by showing that (i) $L(\X^i)$ decreases monotonically, and (ii)  every accumulation point of $(\X^i)_{i\in\mathbb{N}}$ is a stationary point (as defined by \eqref{eq:optimality_cond}) of the local problem associated with each node $q$ \eqref{eq:local_problem} . We then leverage the main result of \cite{musluoglu2022unified_p2, hovine2024distributed} to show optimality with regards to the global optimization problem \eqref{eq:compact_problem}.

    \begin{lemma}[Monotonic decrease of the objective]
        \label{lem:mon_decrease}
        Let $(\X^i)_{i\in\mathbb{N}}$ be a sequence generated by the DASF Algorithm
        with inexact local solver. Then $(L(\X^i))_{i\in\mathbb{N}}$ is a  monotonically decreasing convergent sequence.
    \end{lemma}
    \begin{proof}
        Since by the definition \eqref{eq:local_objective} $\tL_q^i(\lX) = L(C^i_q\lX)$, from the second property of the local solver \eqref{eq:second_prop_ls}, we have 
        \begin{equation}
            L(C_q^i\lX^{i,j})-L(C_q^i\lX^{i,j+1})\geq
            R^i\norm{\lX^{i,j}-\lX^{i,j+1}}^2_F \geq 0.
        \end{equation}
        Summing from $j=0$ to $j=n_i$ and telescoping the left-hand side yields
        \begin{equation}
            L(\X^{i})-L(\X^{i+1})\geq 0,
        \end{equation}
        where we used \eqref{eq:first_prop_ls} to substitute $\X^i=C_q^i\lX^{i,0}$ and the definition of the inexact algorithm to substitute $\X^{i+1} = C_q^i\lX^{i,n_i}$, hence
        proving the monotonic decrease. 

        As $L$ is continuous over its domain with compact sub-level sets, it has a minimum value, and the sequence $(L(\X^i))_{i\in\mathbb{N}}$ is therefore bounded-below \cite{Charalambos2013} and must converge \cite{rudin1976principles}.
    \end{proof}

    \begin{lemma}[Convergence of the Sequence of Residuals]
        \label{lem:residuals}
        Let $(\X^i)_{i\in\mathbb{N}}$ be a sequence generated by the DASF Algorithm
        with inexact local solver, and let $(\lX^{i,j})_{j\in\mathbb{N}}$ be the subsequence generated by the local solver at each iteration $i$. Then ,
        \begin{equation}
            \label{eq:cons_lem_res}
            \lim_{i\to\infty}\norm{\lX^{i,j}-\lX^{i,j+1}}^2_F = 0\quad \forall j
        \end{equation}
        and
        \begin{equation}
            \label{eq:lem_res_p2}
            \lim_{i\to\infty}\norm{\X^i-\X^{i+1}}_F=0.
        \end{equation}
    \end{lemma}
    \begin{proof}
        Using the fact that 
        \begin{equation}
            \label{X0ni}
            \X^i=C^i_q\lX^{i,0}\textrm{ and }\X^{i+1}=C^i_q\lX^{i,n_i},
    \end{equation}
         we have 
        \begin{equation}
            L(\X^i)-L(\X^{i+1}) = \sum_{j=0}^{j=n_i} L(C_q^i\lX^{i,j})-L(C_q^i\lX^{i,j+1}),
        \end{equation}
        From the second property of the local solver \eqref{eq:second_prop_ls}, for any $i$ and $j\leq n_i$
        \begin{equation}
            L(C_q^i\lX^{i,j})-L(C_q^i\lX^{i,j+1})\geq 0,
        \end{equation}
        it must be that, for any $i$ and $j\leq n_i$
        \begin{equation}
            L(\X^i)-L(\X^{i+1})\geq L(C_q^i\lX^{i,j})-L(C_q^i\lX^{i,j+1}),
        \end{equation}
        and hence from \eqref{eq:second_prop_ls} again
        \begin{equation}
            \label{eq:upr_bound_x_step}
            L(\X^i)-L(\X^{i+1})\geq R^i\norm{\lX^{i,j}-\lX^{i,j+1}}^2_F.
        \end{equation}
        As a consequence of Lemma \ref{lem:mon_decrease}, the convergence of $(L(\X^i))_{i\in\mathbb{N}}$ in combination with \eqref{eq:upr_bound_x_step} implies that, for any $0\leq j< n_i$,
        \begin{equation}
            \label{eq:conv_local_steps}
            \lim_{i\to\infty}  L(\X^i)-L(\X^{i+1}) = \lim_{i\to\infty}R^i\norm{\lX^{i,j}-\lX^{i,j+1}}^2_F = 0.
        \end{equation}
        From property \eqref{eq:fourth_prop_ls} of the local solver, the sequence $R^i$ is eventually lower-bounded, and therefore 
        \begin{equation}
            \label{eq:conv_local_steps}
            \lim_{i\to\infty}\norm{\lX^{i,j}-\lX^{i,j+1}}^2_F = 0\quad\forall j,
        \end{equation}
        proving the first statement.

        From \eqref{X0ni} and the triangle inequality, we have
        \begin{multline}
            \label{eq:upper_bound_res}
            \norm{\X^i-\X^{i+1}}_F \leq 
            \sum_{j=0}^{n_i-1} \norm{C_q^i\lX^{i,j}-C_q^i\lX^{i,j+1}}_F\leq \\
            \norm{C_q^i}_F\sum_{j=0}^{n_i-1} \norm{\lX^{i,j}-\lX^{i,j+1}}_F,
        \end{multline}
        where the second inequality relies on the sub-multiplicative property of the Frobenius norm. As $C_q^i=C_q(\X^i)$, and as the sub-level sets of $L(\X^0)$ are compact, $\norm{X^i}_F$ is bounded and by the continuity of the map $C_q(\X)$, so is $\norm{C_q^i}$.  \eqref{eq:conv_local_steps} and \eqref{eq:upper_bound_res} therefore imply that 
        \begin{equation}
            \lim_{i\to\infty}\norm{\X^i-\X^{i+1}}_F=0,
        \end{equation}
        completing the proof.
    \end{proof}

    \begin{lemma}[Node-specific Optimality of Accumulation Points]
        \label{lem:local_optimality}
        Let $(\X^i)_{i\in\mathbb{N}}$ be a sequence generated by the DASF Algorithm
        with inexact local solver. Then, every accumulation point $\lbar{\X}$ of $(\X^i)_{i\in\mathbb{N}}$ is a stationary point of the local problem \eqref{eq:local_problem} for any updating node $q$. 
    \end{lemma}
    \begin{proof}

        By the third and fourth property of the local solver \eqref{eq:third_prop_ls}-\eqref{eq:fourth_prop_ls}, we can find  sequences $(W^i)_{i\in\mathbb{N}}$ and $(c^i)_{i\in\mathbb{N}}$  (with $c^i > 0$) such that 
        \begin{equation}
            \label{eq:incl_wi}
            W^i\in\partial  \tL^i_q(\lX^{i,n_i}),
        \end{equation}
        and 
        \begin{equation}
            \label{eq:upd_subgrad}
            c\norm{\lX^{i,n_i-1}-\lX^{i,n_i}}_F\geq \norm{W^i}_F,
        \end{equation}
        where $c = \limsup_i c^i$.
        We can express equation \eqref{eq:incl_wi} in the global domain of $\X$ by applying the generalized chain rule\footnote{We denote $\partial_{\lX}(\cdot)$ the set of subgradients of the expression in argument, interpreted as a function of the local variable $\lX$.}\footnote{We define the product between a set and a matrix as the set whose elements are the elements of the original set multiplied by that matrix.} \cite{rockafellar2009variational} to $L(C_q^i\lX^{i,n_i})$:
        \begin{multline}
            \label{eq:chain_rule}
            \partial\tL^i_q(\lX^{i,n_i}) =\partial_{\lX} \left(L(C_q^i\lX^{i,n_i})\right) \\= C_q^{iT}\partial L(C_q^i\lX^{i,n_i})= C_q^{iT}\partial L(\X^{i+1})
        \end{multline}
        where we used the fact that $C_q^i\lX^{i,n_i}=\X^{i+1}$ by definition of the algorithm.

        If $\lbar{\X}$ is an accumulation point, then there is (by definition) some index set $\mc{I}$ such that\footnote{$\lim_{i\in\mc{I}\to\infty} a_i$ is a short-hand notation for $\lim_{j\to\infty} a_{i_j}$ where $i_j$ is the $j$-th element in the ordered index set $\mc{I}$.} $\lim_{i\in\mc{I}\to\infty} \X^{i} = \accX$. Furthermore, from Lemma \ref{lem:residuals} and in particular \eqref{eq:lem_res_p2}, we also have 
        \begin{equation}
            \lim_{i\in\mc{I}\to\infty} \X^{i+1} =\lim_{i\in\mc{I}\to\infty} \X^{i}= \accX.
        \end{equation}
        As the number of nodes is finite, we can furthermore select $\mc{I}$ such that the sequence of updating nodes $(q^i)_{i\in\mc{I}}$ is a constant sequence, i.e. we only consider a single updating node $q$. From the continuity of the map $C_q(\X)$ it must also be that there is some $\lbar{C}_q\triangleq C_q(\accX)$ such that
        \begin{equation}
            \lim_{i\in\mc{I}\to\infty} C_q^i = \lim_{i\in\mc{I}\to\infty} C_q(\X^i) ={C}_q(\accX)=\lbar{C}_q.
        \end{equation}

        From \eqref{eq:incl_wi}, \eqref{eq:chain_rule} and the outer-semicontinuity of the set of subgradients \cite{rockafellar2009variational}, we have 
        \begin{equation}
            \label{eq:app_lsc}
            W^i\in C^{iT}_q\partial L(\X^{i+1}) \;\forall i \Rightarrow
            \lim_{i\in\mc{I}\to\infty} W^i\in \lbar{C}^T_q\partial L(\accX).
        \end{equation}
        From \eqref{eq:cons_lem_res} and \eqref{eq:upd_subgrad}, it must be that
        \begin{equation}
            \lim_{i\to\infty} \norm{W^i}_F=0,
        \end{equation}
        and thus  
        \begin{equation}
            \label{eq:above}
            \lim_{i\in\mc{I}\to\infty} W^i = 0.
        \end{equation}
        Combining \eqref{eq:above} with \eqref{eq:app_lsc} yields
        \begin{equation}
            0\in {C_q(\accX)}^T\partial L(\accX)
        \end{equation}
        which is the condition for $\accX$ to be a stationary point of the local problem \eqref{eq:local_problem}, based on the argument in \eqref{eq:chain_rule}. Note that the local problem \eqref{eq:local_problem} is equal to \eqref{eq:compact_problem} equipped with the additional constraint \eqref{eq:param_constraint}, i.e. $\lbar{\X}$ is a stationary point of 
        $$\min_{\X} L(\X) \quad\st \X = C_q(\accX)\lX,\; \lX\in\mathbb{R}^{((K-1)Q+M_q)\times Q}.$$

        We have shown that an accumulation point is a stationary point for at least a single node $q$, it remains to show that it is the case for any $q$. From Lemma \ref{lem:residuals} and in particular \eqref{eq:lem_res_p2},  if $\lbar{\X}$ is an accumulation point of $(\X^{i+1})_{i\in\mathcal{I}}$, then it is also an accumulation point of $(\X^{i+n})_{i\in\mathcal{I}}$ for any $n\geq 0$. Indeed,
        \begin{equation}
            \X^{i+n} = \X^i - (\X^i-\X^{i+1}) - \dots - (\X^{i+n-1}-\X^{i+n}),
        \end{equation}
        and therefore 
        \begin{multline}
            \lim_{i\to\infty} \X^{i+n} - \X^i =\\  \lim_{i\to\infty} (\X^i-\X^{i+1}) - \dots - (\X^{i+n-1}-\X^{i+n}) = 0,
        \end{multline}
        and finally
        \begin{equation}
            \lim_{i\in\mc{I}\to\infty} \X^{i} = \lim_{i\in\mc{I}\to\infty} \X^{i+n} = \accX.
        \end{equation}
        As we have shown that $\lbar{\X}$ is a stationary point for at least one node, it is a stationary point for all nodes as the sequence $(\X^{i+n})_{i\in\mathcal{I}}$ will be associated with node\footnote{assuming a sequential selection rule, but the conclusion stays valid as long as every node is selected infinitely many times.} $(q+n) \mod K$.
    \end{proof}

    Lemma \ref{lem:local_optimality} asserts that any accumulation point of DASF is a stationary point of all local problems at all nodes. We will use this result to show that these accumulation points are also stationary points of the centralized problem \eqref{eq:generic_problem}-\eqref{eq:generic_problem_ns}. For this, we need a technical condition that is usually satisfied with high probability if the number of constraints is not too high (see below). The condition can actually be viewed as a compressed version of the standard linear independence constraint qualifications (LICQ) in numerical optimization \cite{nocedal1999numerical},  except that they apply to the ``compressed'' gradients rather than the actual gradients of the constraint functions.

    \begin{proposition}[Compressed LICQ]
        \label{prop:compressed_licq}
        Let $\vartheta_j:\X\mapsto\eta_j(\X^T\y,\X^T\C_j)$ in the smooth case, and  $\vartheta_j:\X\mapsto\eta_j(\X_k^T\y_k,\X_k^T\C_{j,k})$ in the non-smooth case, and define the active constraint set $\mc{A}(\X)\triangleq \{j\in\mc{J}_I\;|\; \vartheta_j(\X)=0\}$. If the elements of the set 
        \begin{equation}
            \label{eq:compressed_licq}
            \{\blkd(\accX_1, \dots, \accX_K)^T\nabla \vartheta_j(\accX) \;|\; j\in\mc{A}(\accX)\cup\mc{J_E}\}
        \end{equation}
        are linearly independent matrices and $\accX$ is a stationary point of the local problems  for any node $q$, i.e.
        \begin{equation}
            \forall q\in\mc{K},\;0\in \lbar{C}^T_q\partial L(\accX),
        \end{equation}
        then it is also a stationary point of the global problem \eqref{eq:generic_problem}-\eqref{eq:generic_problem_ns}, i.e.
        \begin{equation}
            0\in\partial L(\accX).
        \end{equation}
    \end{proposition}
    A proof of the above statement is available in \cite{musluoglu2022unified_p2} for the smooth case and in \cite{hovine2024distributed} for the non-smooth case. Note that this qualification is automatically violated if the number of active constraints exceeds $KQ^2$ (which is the dimension of the vector space in which the compressed gradients live). See \cite{musluoglu2022unified_p1, musluoglu2022unified_p2} for details and an extension of this condition to the case of arbitrary network topologies.

    \begin{theorem}[Convergence and Optimality]
        \label{th:inexact}
        Let $(\X^i)_{i\in\mathbb{N}}$ be a sequence generated by the DASF algorithm with inexact solver. Then if Assumption \ref{ass:main} is satisfied and the qualification \eqref{eq:compressed_licq} is also satisfied at the accumulation points of $(\X^i)_{i\in\mathbb{N}}$,  $(\X^i)_{i\in\mathbb{N}}$ converges to the set of stationary points of the global problem \eqref{eq:generic_problem}-\eqref{eq:generic_problem_ns}, that is
        \begin{equation}
            \lim_{i\to\infty}\min_{X\in\Omega}\norm{X-\X^i}_F =0,
        \end{equation}
        where $\Omega$ denotes the set of stationary points of problem \eqref{eq:generic_problem}-\eqref{eq:generic_problem_ns}.

        Furthermore, if the number of stationary points of  \eqref{eq:generic_problem}-\eqref{eq:generic_problem_ns} is finite, then $(\X^i)_{i\in\mathbb{N}}$ converges to a single point.
    \end{theorem}
    \begin{proof}
        Under Assumption \ref{ass:main}, Lemma \ref{lem:local_optimality} asserts that  any accumulation point is a stationary point of the local problem at every node. Proposition \ref{prop:compressed_licq} ensures that under the qualification \eqref{eq:compressed_licq}, the local stationarity at every node implies stationarity for the global problem, proving the first part of the theorem (i.e., convergence to the set of stationary points). The second part (convergence to a single point) is a direct consequence of Lemma \ref{lem:residuals} and is a standard analysis result. See the proof of \cite[Theorem 4]{musluoglu2022unified_p2} for details.
    \end{proof}

    As a consequence of the above theorem, interleaving the steps of DASF with a single descent step is sufficient to ensure convergence to a stationary point.

    \subsection{A Bound on the Convergence Rate}

    Obtaining a convergence rate for the objective value would unfortunately require much more assumptions on the functions involved in \eqref{eq:generic_problem}-\eqref{eq:generic_problem_ns}, than what we assumed so far. Still, we can obtain a lower bound on the convergence rate of a local optimality measure at each node, that is 
    \begin{equation}
        \label{eq:wi}
        w^i\triangleq\dist (0,\partial\tL^i_{q^i}(\lX^{i,n_i}))\triangleq \min_{W\in \partial\tL^i_{q^i}(\lX^{i,n_i})} \norm{W}.
    \end{equation}
    As our actual objective is to find stationary points rather than an actual minimum, $w^i$ is an appropriate measure of optimality. We could also have considered the step-length $\norm{X^i-X^{i+1}}$ as a measure of optimality, which would yield a rate similar to the one presented hereafter. 

    \begin{proposition}[Convergence Rate Bound]
        Let $(\X^i)_{i\in\mathbb{N}}$ be a sequence generated by the DASF algorithm with inexact solver, and let $w^i$ be as defined in \eqref{eq:wi}. Then there is some positive constant $a$ such that
        \begin{equation}
            \label{eq:rate}
            \min_{j\leq i} w^j \leq a \frac{\sqrt{L(\X^0)-L^\star}}{\sqrt{i+1}},
        \end{equation}
        where $L^\star$ is the minumum value of $L$.
    \end{proposition}
    \begin{proof}
        This follows from a well-known proof argument \cite{beck2017first, themelis2018proximal}.
        From \eqref{eq:upr_bound_x_step} and condition \eqref{eq:third_prop_ls}, we have 
        \begin{equation}
            L(\X^i)-L(\X^{i+1}) \geq \frac{R^i}{c^i}(w^i)^2.
        \end{equation}
        Summing the inequality from $0$ to $i$ and telescoping yields
        \begin{equation}
            L(\X^0)-L(\X^{i+1}) \geq \sum_{j=0}^i \frac{R^j}{c^j}(w^j)^2.
        \end{equation}
        From condition \eqref{eq:fourth_prop_ls}, there is some $r>0$ such that 
        \begin{equation}
            \liminf_{i\to\infty} \frac{R^i}{c^i} =r ,
        \end{equation}
        and hence, using the fact that $L^\star\leq L(\X^{i+1})$,
        \begin{equation}
            L(\X^0)-L^\star \geq r \sum_{j=0}^i(w^j)^2 \geq r(i+1)\min_j (w_j)^2,
        \end{equation}
        which can be rearranged to obtain \eqref{eq:rate} with $a=r^{-1/2}$.

    \end{proof}
    For a given level of optimality $\min w^i$, the bound gives a linear relationship between the required number of iterations and the current error $L(\X^i)-L^\star$. In other words, if the algorithm is set to stop after a sufficiently small norm of the subgradient is reached, we can guarantee that the algorithm will run for a number of iterations (hence time) that is at most proportional to the current error $L(\X^i)-L^\star$.

\section{Numerical Experiments}
\label{sec:sims}

\begin{figure*}
\includegraphics[width=\textwidth]{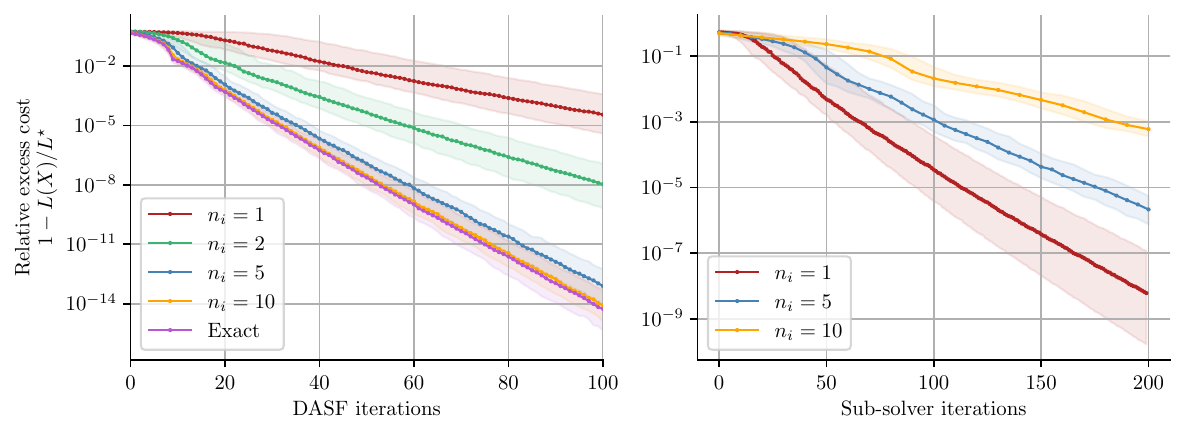}
\caption{Convergence of the DASF algorithm with inexact solver for the Max-SNR problem. Each curve was generated by 1000 Monte-Carlo runs. Solid curves depict the median trajectories, while the shaded areas of corresponding color depict the the 5\%-95\% percentiles regions.}
\label{fig:sim}
\end{figure*}

In order to qualitatively demonstrate the convergence of the DASF algorithm with inexact solver, we consider an SNR maximization (Max-SNR) problem \cite{van1988beamforming} defined as 
\begin{equation}
    \label{eq:sim_prob}
    \max_{\x} \E{\|\x^T\y(t)\|^2}\quad \st \E{\|\x^T\bm{n}(t)\|_F^2} =  1
\end{equation}
where $x$ denotes the $M\times 1$ spatial filter ($Q=1$), $\y(t)$ is modelled as $\y(t) = \bm{a}{d}(t) + \bm{n}(t)$, $\bm{n}(t)$ is an $M$-dimensional i.i.d white Gaussian noise signal, $d(t)$ is a single-channel (unknown) target signal, and where $\bm{a}$ is some unknown mixing vector with entries sampled from $\mc{N}(0,1)$. The purpose of the multi-input single-output (MISO) filter $\x$ is to maximize the SNR of $d$ in the output filtered signal $z(t)=x^T\y(t)$. It is assumed that the nodes of the WSN are able to collect samples of both $\y(t)$ and $\bm{n}(t)$. For our simulations we set $M=100$, $M_k=10$, $K=10$, i.e., there are 10 nodes, with each node collecting 10 different channels of the 100-channel signals $\y(t)$ and $\bm{n}(t)$. The noise variance is 10, while the signal variance is 1. Our figure of merit is the relative excess cost, defined as
\begin{equation}
    1-\frac{L(\X)}{L^\star},
\end{equation}
where $L^\star$ denotes the optimal objective value of \eqref{eq:sim_prob}. Note that the solution of \eqref{eq:sim_prob} can be solved by computing a generalized eigenvalue decomposition (GEVD) on the covariance matrices $\E{\y(t)\y(t)^T}$ and $\E{\bm{n}(t)\bm{n}(t)^T}$ \cite{ghojogh2019eigenvalue}. 

Figure  \ref{fig:sim} depicts  the convergence of the inexact version of DASF applied to problem \eqref{eq:sim_prob} with different numbers of sub-solver iterations ($n^i$). The exact (GEVD) solver uses Cuppen's divide and conquer algorithm \cite{1999lapack}, while the inexact solver relies on the generalized power method \cite{kozlowski1992power}.  On the left part of the figure, we see that the algorithm converges to the optimal objective value with a single sub-solver iteration per-node, and that 10 sub-solver iterations are sufficient to reach a convergence rate at par with the exact solver. Although, based on this first plot,  one might conclude that more  sub-solver iterations is better, the right part of Figure \ref{fig:sim} tells a very different story. The same figure of merit is depicted, but in terms of the total number of sub-solver iterations, rather than ``global'' DASF iterations. This second plot therefore depicts the algorithm's convergence as a function of ``real'' time and offers a fairer comparison. Indeed, the sub-solver with the lowest number of iterations ($n_i=1$) outperforms all the other ones. It seems that the smaller progress achieved at each step is largely compensated by the fact that the rate at which the updating node role is passed-on is much larger, resulting in the network being able to update the filter in more ``directions'' (i.e. with different subspace constraints \eqref{eq:param_constraint}) in a given time-interval, albeit with smaller progress at each updating node. These more frequent updates of the updating node do come with a cost, as a new batch of $N$ compressed samples data should normally be retransmitted every time the updating node changes. This can be largerly compensated by applying the data-reuse technique described in \cite{musluoglu2022improved}
, making the DASF algorithm with inexact local solver an excellent option for adaptively tracking a spatial filter.

\section{Conclusion}
\label{sec:concl}
We have extended the convergence results of the DASF algorithm to a setting where the exact solver is replaced by an inexact, iterative one. We have shown, both by the means of proofs and simulations, that the convergence properties of DASF were preserved under more relaxed technical conditions. Furthermore, the simulations have shown that DASF used with an iterative solver could display better convergence properties, at the cost of a slighly increased bandwidth, assuming a data-reuse scheme as in \cite{musluoglu2022improved}.

\appendix

\subsection{Common Inexact Solvers}
\label{apx:common_solvers}

We describe in this section a couple of methods satisfying the requirements set by Definition \ref{def:inexact_solver}. In what follows, we consider methods to minimize a real-valued function $h:\R^{M\times Q}\mapsto \R$, whose precise structure (e.g. smoothness, convexity, etc.) will be redefined for every method. We first state a useful definition:
\begin{definition} A smooth function $h:\R^{M\times Q}\mapsto \R$ is said to have $R$-Lipschitz gradients if there is some $R>0$ such that for any $X,Y\in\R^{M\times Q}$
\begin{equation}
    \label{eq:lipshitz}
    \norm{\nabla h(X) - \nabla h(Y)}_F \leq R\norm{X-Y}_F.
\end{equation}
In addition,\cite[Descent Lemma]{beck2017first}, \eqref{eq:lipshitz} also implies that
\begin{equation}
    \label{eq:descent_lemma}
    h(Y) \leq h(X) + \langle \nabla h(X), Y-X\rangle_F + \frac{R}{2}\norm{X-Y}_F^2,
\end{equation}
where $\langle A, B \rangle_F\triangleq \tr{A^TB}$.
\end{definition}

\subsubsection{Line-Search Methods: Gradient Descent and Newton's Method}

Line search methods are defined by the update rule
\begin{equation}
    \label{eq:update_rule_line_search}
    \X^{i+1} = \X^i + \mu^i P^i,
\end{equation}
where $P^i$ is the \emph{search direction} and $\mu^i$ is an appropriately chosen step-size. 

\paragraph{Gradient Descent}
In the case of Gradient Descent, we have $P^i=-\nabla h(\X^i)$ and hence
\begin{equation}
    \label{eq:update_rule_gradient}
 X^{i+1}-X^i=-\mu^i\nabla h(\X^i).
\end{equation}
As a consequence, we have 
\begin{equation}
    \frac{1}{\inf_i\mu^i}  \norm{X^{i+1}-X^i}\geq\norm{\nabla h(\X^i)}.
\end{equation}
As $h$ is smooth, $\partial h(\X^i) = \{\nabla h(\X^i)\}$ \cite{rockafellar2009variational}, and condition \eqref{eq:third_prop_ls} of Definition \ref{def:inexact_solver} is therefore satisfied.

 Applying the Descent Lemma \eqref{eq:descent_lemma}  to $\X^i,\X^{i+1}$  we have
\begin{multline}
    \label{eq:descent_lemma_line_search}
    h(\X^i)-h(\X^{i+1})\\ \geq - \langle \nabla h(\X^i), X^{i+1}-X^i\rangle_F - \frac{R}{2}\norm{X^{i+1}-X^i}_F^2.
\end{multline}
From \eqref{eq:update_rule_gradient}, \eqref{eq:descent_lemma_line_search} becomes
\begin{equation}
    h(\X^i)-h(\X^{i+1}) \geq \left( \frac{1}{\mu^i} - \frac{R}{2}\right) \norm{X^{i+1}-X^i}_F^2.
\end{equation} and condition \eqref{eq:second_prop_ls} is thus also satisfied if $\sup_i \mu^i < \frac{2}{R}$.

\paragraph{Newton's Method}
Newton's step consists in setting $P^i = -\left(\nabla^2 h(\X^i)\right)^{-1}\nabla h(\X^i)$. Hence, from the update rule \eqref{eq:update_rule_line_search}, we have
\begin{equation}
    \label{eq:update_rule_newton}
    -\frac{1}{\mu^i}\nabla^2 h(\X^i) (\X^{i+1}-\X^i) = \nabla h(\X^i).
\end{equation}
Taking the norm on both sides and denoting as $\lambda_{\max}$ an upper bound on the largest eigenvalue of the Hessian (across all iterations),
\begin{equation}
    \label{eq:newt_upd_2}
    \frac{\lambda_{\max}}{\inf_i\mu^i}  \norm{X^{i+1}-X^i}\geq\norm{\nabla h(\X^i)},
\end{equation}
and condition \eqref{eq:third_prop_ls} is fulfilled by the same reasoning as for gradient descent.

Taking the inner product of both sides of \eqref{eq:update_rule_newton} with $(\X^{i+1}-\X^i)$ and denoting the Cholesky factorization of the Hessian as $\nabla^2 f(\X^i)\triangleq UU^T$, we have 
\begin{equation}
 -\langle \nabla h(\X^i), X^{i+1}-X^i\rangle_F  =  \frac{1}{\mu^i}\norm{ U^T (\X^{i+1}-\X^i)}^2.
\end{equation}
Hence 
\begin{equation}
    -\langle \nabla h(\X^i), X^{i+1}-X^i\rangle_F  \geq  \frac{\lambda_{\min}}{\sup_i\mu^i}\norm{  \X^{i+1}-\X^i}^2
\end{equation}
where $\lambda_{\min}$ denotes a lower bound on the smallest eigenvalue of the Hessian (across all iterations). Combining the above inequality with \eqref{eq:descent_lemma_line_search} yields
\begin{equation}
    h(\X^i)-h(\X^{i+1}) \geq \left(  \frac{\lambda_{\min}}{\sup_i\mu^i} - \frac{R}{2}\right) \norm{X^{i+1}-X^i}_F^2,
\end{equation}
and condition \eqref{eq:second_prop_ls} is satisfied as long as the step-size satisfies ${\sup_i\mu^i} < \frac{2\lambda_{\min}}{R}$ (note that this also implies that the Hessian must be positive definite).

In practice, the bounds on the step-size cannot be known apriori, but the same results can be obtained if a backtracking line-search  \cite{nocedal1999numerical, beck2017first, armijo1966minimization} is used to obtain the step-size (we omit the analysis for brevity).


\subsubsection{Regularized Exact Solver}
Given some previous estimate of the solution $\X^-$, a regularized exact solver would update the solution as (we drop any iteration index, as the solver only performs a single iteration)
\begin{equation}
    \label{eq:reg_solver}
    \X^+ = \argmin_X h(\X) +\mu \norm{\X-\X^-}_F^2
\end{equation}
instead of $\min_X h(\X)$, where $\mu>0$ is a freely chosen parameter. From the optimality of $\X^+$ in \eqref{eq:reg_solver}, we have
\begin{equation}
    h(\X^-) -  h(\X^+)\geq\mu \norm{\X^+-\X^-}_F^2,
\end{equation}
and condition \eqref{eq:second_prop_ls} is satsfied. Furthermore, the optimality of $\X^+$ also implies that (using the sum-rule for subgradients \cite{Royset2021})
\begin{equation}
    0\in\partial h(\X^+) + 2\mu (X^+ - X^-),
\end{equation}
And thus there is $W\triangleq 2\mu (X^- - X^+)\in\partial h(\X^+)$, and condition \eqref{eq:third_prop_ls} is trivially satisfied.

\subsubsection{Power Method}
We will rely on the proof available in \cite{attouch2013convergence} that the projected gradient algorithm satisfies conditions \eqref{eq:second_prop_ls} and \eqref{eq:third_prop_ls} of Definition \ref{def:inexact_solver}, even when the constraint set is non-convex. The projected gradient algorithm is defined by
\begin{equation}
    \label{eq:update_projected_gradient}
    X^{i+1} = P_\mc{C}\left(X^i-\mu \nabla h(X^i)\right)
\end{equation}
where $P_\mc{C}$ is the projection on some constraint set $\mc{C}$ and $\mu$ some step-size. 
We will show that the power method is a particular case of \eqref{eq:update_projected_gradient}. The power method finds the eigenvector associated with the largest eigenvalue of some matrix $A$. It therefore finds a solution of 
\begin{equation}
    \label{eq:quad_min}
    \min_x -x^TAx\quad \st x^Tx = 1.
\end{equation}
The power method's udpate rule is 
\begin{equation}
    \label{eq:update_rule_power}
    x^{i+1} = \frac{Ax^i}{\norm{Ax^i}}.
\end{equation}
Let us now assume that we solve the following (equivalent) problem using a projected gradient algorithm:
\begin{equation}
    \min_x -x^T\frac{(A-I)}{2\mu}x\quad \st x^Tx = 1.
\end{equation}
Denoting the objective $h$, its gradient is
\begin{equation}
    \nabla h(x) = \frac{(I-A)}{\mu}x,
\end{equation}
The projected gradient method with step-size $\mu$ for this problem is thus
\begin{equation}
    x^{i+1} = P_{\mc{C}}\left(x^i - \mu \frac{x^i - Ax^i}{\mu}\right) =  P_{\mc{C}}\left(Ax^i\right)
\end{equation}
with $\mc{C}$ denoting in this case the unit ball, which is equivalent to the update rule of the power method \eqref{eq:update_rule_power}. Note that we skipped the discussion on the condition that the step-size should be smaller than the inverse of the Lipschitz constant of $\nabla h(x)$ \cite{attouch2013convergence}, and simply mention that this condition can always be enforced by applying the proper scaling on $A$, as this changes neither the solution of \eqref{eq:quad_min} nor the update rule \eqref{eq:update_rule_power}.

\bibliographystyle{IEEEtran}
\bibliography{main_biblio}
\end{document}